\documentclass{llncs}
\usepackage{amssymb}
\usepackage{pxfonts}
\newdimen\proofrulebreadth \proofrulebreadth=.05em
\newdimen\proofdotseparation \proofdotseparation=1.25ex
\newdimen\proofrulebaseline \proofrulebaseline=2ex
\newcount\proofdotnumber \proofdotnumber=3
\let\then\relax
\def\hfi{\hskip0pt plus.0001fil}
\mathchardef\squigto="3A3B
\newif\ifinsideprooftree\insideprooftreefalse
\newif\ifonleftofproofrule\onleftofproofrulefalse
\newif\ifproofdots\proofdotsfalse
\newif\ifdoubleproof\doubleprooffalse
\let\wereinproofbit\relax
\newdimen\shortenproofleft
\newdimen\shortenproofright
\newdimen\proofbelowshift
\newbox\proofabove
\newbox\proofbelow
\newbox\proofrulename
\def\shiftproofbelow{\let\next\relax\afterassignment\setshiftproofbelow\dimen0 }
\def\shiftproofbelowneg{\def\next{\multiply\dimen0 by-1 }%
\afterassignment\setshiftproofbelow\dimen0 }
\def\setshiftproofbelow{\next\proofbelowshift=\dimen0 }
\def\setproofrulebreadth{\proofrulebreadth}

\def\prooftree{
\ifnum  \lastpenalty=1
\then   \unpenalty
\else   \onleftofproofrulefalse
\fi
\ifonleftofproofrule
\else   \ifinsideprooftree
        \then   \hskip.5em plus1fil
        \fi
\fi
\bgroup
\setbox\proofbelow=\hbox{}\setbox\proofrulename=\hbox{}%
\let\justifies\proofover\let\leadsto\proofoverdots\let\Justifies\proofoverdbl
\let\using\proofusing\let\[\prooftree
\ifinsideprooftree\let\]\endprooftree\fi
\proofdotsfalse\doubleprooffalse
\let\thickness\setproofrulebreadth
\let\shiftright\shiftproofbelow \let\shift\shiftproofbelow
\let\shiftleft\shiftproofbelowneg
\let\ifwasinsideprooftree\ifinsideprooftree
\insideprooftreetrue
\setbox\proofabove=\hbox\bgroup$\displaystyle 
\let\wereinproofbit\prooftree
\shortenproofleft=0pt \shortenproofright=0pt \proofbelowshift=0pt
\onleftofproofruletrue\penalty1
}

\def\eproofbit{
\ifx    \wereinproofbit\prooftree
\then   \ifcase \lastpenalty
        \then   \shortenproofright=0pt  
        \or     \unpenalty\hfil         
        \or     \unpenalty\unskip       
        \else   \shortenproofright=0pt  
        \fi
\fi
\global\dimen0=\shortenproofleft
\global\dimen1=\shortenproofright
\global\dimen2=\proofrulebreadth
\global\dimen3=\proofbelowshift
\global\dimen4=\proofdotseparation
\global\count255=\proofdotnumber
$\egroup  
\shortenproofleft=\dimen0
\shortenproofright=\dimen1
\proofrulebreadth=\dimen2
\proofbelowshift=\dimen3
\proofdotseparation=\dimen4
\proofdotnumber=\count255
}

\def\proofover{
\eproofbit 
\setbox\proofbelow=\hbox\bgroup 
\let\wereinproofbit\proofover
$\displaystyle
}%
\def\proofoverdbl{
\eproofbit 
\doubleprooftrue
\setbox\proofbelow=\hbox\bgroup 
\let\wereinproofbit\proofoverdbl
$\displaystyle
}%
\def\proofoverdots{
\eproofbit 
\proofdotstrue
\setbox\proofbelow=\hbox\bgroup 
\let\wereinproofbit\proofoverdots
$\displaystyle
}%
\def\proofusing{
\eproofbit 
\setbox\proofrulename=\hbox\bgroup 
\let\wereinproofbit\proofusing
\kern0.3em$
}

\def\endprooftree{
\eproofbit 
  \dimen5 =0pt
\dimen0=\wd\proofabove \advance\dimen0-\shortenproofleft
\advance\dimen0-\shortenproofright
\dimen1=.5\dimen0 \advance\dimen1-.5\wd\proofbelow
\dimen4=\dimen1
\advance\dimen1\proofbelowshift \advance\dimen4-\proofbelowshift
\ifdim  \dimen1<0pt
\then   \advance\shortenproofleft\dimen1
        \advance\dimen0-\dimen1
        \dimen1=0pt
        \ifdim  \shortenproofleft<0pt
        \then   \setbox\proofabove=\hbox{%
                        \kern-\shortenproofleft\unhbox\proofabove}%
                \shortenproofleft=0pt
        \fi
\fi
\ifdim  \dimen4<0pt
\then   \advance\shortenproofright\dimen4
        \advance\dimen0-\dimen4
        \dimen4=0pt
\fi
\ifdim  \shortenproofright<\wd\proofrulename
\then   \shortenproofright=\wd\proofrulename
\fi
\dimen2=\shortenproofleft \advance\dimen2 by\dimen1
\dimen3=\shortenproofright\advance\dimen3 by\dimen4
\ifproofdots
\then
        \dimen6=\shortenproofleft \advance\dimen6 .5\dimen0
        \setbox1=\vbox to\proofdotseparation{\vss\hbox{$\cdot$}\vss}%
        \setbox0=\hbox{%
                \advance\dimen6-.5\wd1
                \kern\dimen6
                $\vcenter to\proofdotnumber\proofdotseparation
                        {\leaders\box1\vfill}$%
                \unhbox\proofrulename}%
\else   \dimen6=\fontdimen22\the\textfont2 
        \dimen7=\dimen6
        \advance\dimen6by.5\proofrulebreadth
        \advance\dimen7by-.5\proofrulebreadth
        \setbox0=\hbox{%
                \kern\shortenproofleft
                \ifdoubleproof
                \then   \hbox to\dimen0{%
                        $\mathsurround0pt\mathord=\mkern-6mu%
                        \cleaders\hbox{$\mkern-2mu=\mkern-2mu$}\hfill
                        \mkern-6mu\mathord=$}%
                \else   \vrule height\dimen6 depth-\dimen7 width\dimen0
                \fi
                \unhbox\proofrulename}%
        \ht0=\dimen6 \dp0=-\dimen7
\fi
\let\doll\relax
\ifwasinsideprooftree
\then   \let\VBOX\vbox
\else   \ifmmode\else$\let\doll=$\fi
        \let\VBOX\vcenter
\fi
\VBOX   {\baselineskip\proofrulebaseline \lineskip.2ex
        \expandafter\lineskiplimit\ifproofdots0ex\else-0.6ex\fi
        \hbox   spread\dimen5   {\hfi\unhbox\proofabove\hfi}%
        \hbox{\box0}%
        \hbox   {\kern\dimen2 \box\proofbelow}}\doll%
\global\dimen2=\dimen2
\global\dimen3=\dimen3
\egroup 
\ifonleftofproofrule
\then   \shortenproofleft=\dimen2
\fi
\shortenproofright=\dimen3
\onleftofproofrulefalse
\ifinsideprooftree
\then   \hskip.5em plus 1fil \penalty2
\fi
}

\def\df{{\stackrel{\mathrm{def}}{\Longleftrightarrow}}}

\newcommand{\pts}{\mathit{pts}}
\newcommand{\ptsp}{\mathit{pts}^\prime}
\newcommand{\lb}{\llbracket}
\newcommand{\rb}{\rrbracket}
\newcommand{\ra}{\rightarrow}
\newcommand{\Ra}{\Longrightarrow}
\newcommand{\while}{\mathit{while\ b\ do\ } S_t}
\newcommand{\ifs}{\mathit{if\ b\ then\ }S_t \mathit{\ else\ }S_f}

\begin{document}
\title{Dead Code Elimination Based Pointer Analysis for Multithreaded Programs}
\author{Mohamed A. El-Zawawy} \institute{College of Computer and
Information Sciences\\
Al-Imam M. I.-S. I. University\\ Riyadh 11432\\Kingdom of Saudi
Arabia\\ and \\ Department of
Mathematics\\ Faculty of Science \\ Cairo University\\ Giza 12613\\
Egypt\\ \email{maelzawawy@cu.edu.eg}} \maketitle
\begin{abstract}
This paper presents a new approach for optimizing multitheaded
programs with pointer constructs. The approach has applications in
the area of certified code (proof-carrying code) where a
justification or a proof for the correctness of each optimization is
required. The optimization meant here is that of dead code
elimination.

Towards optimizing multithreaded programs the paper presents a new
operational semantics for parallel constructs like join-fork
constructs, parallel loops, and conditionally spawned threads. The
paper also presents a novel type system for flow-sensitive pointer
analysis of multithreaded programs. This type system is extended to
obtain a new type system for live-variables analysis of
multithreaded programs. The live-variables type system is extended
to build the third novel type system, proposed in this paper, which
carries the optimization of dead code elimination. The justification
mentioned above takes the form of type derivation in our approach.
\end{abstract}

\footnote{Short title: Dead Code Elimination for Multithreaded
Programs. }
\section{Introduction}\label{intro}

One of the mainstream programming approaches today is
multithreading. Using multiple threads is useful in many ways like
(a) concealing suspension caused by some commands, (b) making it
easier to build huge software systems, (c) improving execution of
programs especially those that are executed on multiprocessors, and
(d) building advanced user interfaces.

The potential interaction between threads in a multithreaded program
complicates both the compilation and the program analysis processes.
Moreover this interaction also makes it difficult to extend the
scope of program analysis techniques of sequential programs to cover
multithreaded programs.

Typically optimizing multithreaded programs is achieved in an
algorithmic form using data-flow analyses. This includes
transforming the given program into a control-flow graph which is a
convenient form for the algorithm to manipulate. For some
applications like certified code, it is desirable to associate each
program optimization with a justification or a proof for the
correctness of the optimization. For these cases, the algorithmic
approach to program analysis is not a good choice as it does not
work on the syntactical structure of the program and hence does not
reflect the transformation process. Moreover the desired
justification must be relatively simple as it gets checked within a
trusted computing base.

Type systems stand as a convenient alternative for the algorithmic
approach of program analyses when a justification is necessary. In
the type systems approach, analysis and optimization of programs are
directed by the syntactical structure of the program. Inference
rules of type systems are advantageously relatively simple and so is
the justification which takes the form of a type derivation in this
case. The adequacy of type systems for program analysis has already
been studied for example in~\cite{Benton04,Laud06,Nielson02}

Pointer analysis is among the most important program analyses and it
calculates information describing contents of pointers at different
program points. The application of pointer analysis to multithreaded
programs results in information that is required for program
analyses and compiler optimizations such as live-variables analysis
and dead code elimination, respectively. The live-variables analysis
finds for each program point the set of variables whose values are
used usefully in the rest of the program. The results of
live-variables analysis is necessary for the optimization of dead
code elimination which removes code that has no effect on values of
variables of interest at the end of the program.

This paper presents a new approach for optimizing multithreaded
programs with pointer constructs. The scope of the proposed approach
is broad enough to include certified (proof-carrying) code
applications where a justification for optimization is necessary.
Type systems are basic tools of the new approach which considers
structured parallel constructs such as join-fork constructs,
parallel loops, and conditionally spawned threads. The
justifications in our approach take the form of type derivations.
More precisely, the paper presents a type system for flow-sensitive
pointer analysis of multithreaded programs. The live-variables
analysis of multithreaded programs is also treated in this paper by
a type system which is an extension of the type system for pointer
analysis. The extension has the form of another component being
added to points-to types. The dead code elimination of multithreaded
programs is then achieved using a type system which is again an
extension of the type system for live-variables analysis. This time
the extension takes the form of a transformation component added to
inference rules of the type system for live-variables analysis. To
prove the soundness of the three proposed type systems, a novel
operational semantics for parallel constructs is proposed in this
paper.
\begin{figure}[thp]
{\footnotesize{
\begin{minipage}[b]{0.2\linewidth}
\quad\end{minipage}
\begin{minipage}[b]{0.2\linewidth}
\begin{tabular}{l}
  $1.\quad x\coloneqq \& y;$ \\
  $2.\quad *x\coloneqq  2;$ \\
  $3.\quad \textit{par}\{$ \\
  $4.\hspace{.8cm}  \{y\coloneqq 4\}$  \\
  $5.\hspace{.8cm}  \{y\coloneqq 5;$\\
  $6.\hspace{.81cm}  *x\coloneqq 6\}$\\
  $7.\hspace{.8cm}\};$  \\
  $8.\quad x\coloneqq 8;$  \\
  $9.\quad x\coloneqq 9$
\end{tabular}
\end{minipage}
\begin{minipage}[b]{0.2\linewidth}
  $\quad\quad\Longrightarrow$
\end{minipage}
\begin{minipage}[b]{0.2\linewidth}
\begin{tabular}{l}
  $\quad x\coloneqq \& y;$ \\
  $\quad \textit{skip};$ \\
  $\quad \textit{par}\{$ \\
  $\hspace{.8cm}  \{y\coloneqq 4\}$  \\
  $\hspace{.8cm}  \{\textit{skip};$\\
  $\hspace{.81cm}  *x\coloneqq 6\}$\\
  $\hspace{.8cm}\};$  \\
  $\quad \textit{skip};$  \\
  $\quad x\coloneqq 9$
\end{tabular}
\end{minipage}
\caption{A motivating example}\label{example1} }}\end{figure}

Some comments are in order. The language that we use does not allow
pointer arithmetic, so it can be thought of as a Pascal-style
language. The parallel construct assumes that each thread is
executed atomically. This is a restriction that is not very common
to many implementation of threads. But based on our construct, it
should be easy to treat a more general construct. Writing computer
programs to efficiently calculate the types, using our type systems,
is not hard. This is so because each time the program calculates a
type, the program will be given a pre-type and seeking its post type
or vice versa. In the semantics our langauge, we assume that
$\mathbb{Z}$ and the set of addresses, Addrs, are disjoint sets.

\subsection*{Motivation}
Figure~\ref{example1} presents a motivating example of the work
presented in this paper. Consider the program on the left-hand-side
of the figure. Suppose that at the end of the program we are
interested in the values of $x$ and $y$. We note that the assignment
in line $8$ is dead code as the variable $x$ is modified in line $9$
before we make any use of the value that the variable gets in line
$8$. The assignment in line $2$ indirectly modifies $y$ which is
modified again in the \textit{par} command before any useful use of
the value that $y$ gets in line $2$. Therefore line $2$ is  dead
code. The \textit{par} command has two threads which can be executed
in any order. If the first thread is executed first then assignments
in lines $4$ and $5$ become dead code. If the second thread is
executed first then assignments in lines $5$ and $6$ become dead
code. Therefore the dead code in the \textit{par} command is the
assignment in line $5$ only.

This paper presents a technique that discovers and removes such dead
code in parallel structured programs with pointer constructs. The
output of the technique is a program like that on the
right-hand-side of Figure~\ref{example1}. In addition to the
join-fork construct (\textit{par}), the paper also considers other
parallel constructs like conditionally spawned threads and parallel
loops. With each such program optimization, our technique presents a
justification or a proof for the correctness of the optimization.
The proof takes the form of a type derivation.

\subsection*{Contributions}
Contributions of this paper are the following:
\begin{enumerate}
\item A simple yet powerful operational semantics for multi-threaded
programs with pointer constructs.
\item A novel type system for pointer analysis of multithreaded
programs. To our knowledge, this is the first attempt to use type
systems for pointer analysis of multithreaded programs.
\item A new type systems for live-variables analysis of
multithreaded programs.
\item An original type system for the optimization of
dead code elimination for multithreaded programs.
\end{enumerate}

\subsubsection*{Organization} The rest of the paper is organized as
follows. The language that we study (a while language enriched with
pointer and parallel constructs) and an operational semantics for
its constructs are presented in Section~\ref{semantics}.
Sections~\ref{pointers} and~\ref{live} present our type systems for
flow-sensitive pointer and live-variables analyses, respectively.
The type system carrying program optimization is introduced in
Section~\ref{optimization}. Related work is discussed in
Section~\ref{rwork}.

\section{Programming language}\label{semantics}
This section presents the programming language (Figure~\ref{lang})
we use together with an operational semantics for its constructs.
The language is the simple \textit{while} language~\cite{Hoare69}
enriched with commands for pointer manipulations and structured
parallel constructs.
\begin{figure}[h]
\begin{eqnarray*}
&  & {n}\in {\mathbb{Z}},\ {x}\in {\textit{Var}},\ \textit{and}\ {\oplus}
\in{\{+,-,\times\}}\\
e \in \textit{Aexprs}  &  \Coloneqq & {x}\mid  {n} \mid {e_1 \oplus e_2}\\
b\in \textit{Bexprs} &  \Coloneqq & {\textit{true}} \mid
{\textit{false}} \mid {\neg b}\mid {e_1 = e_2} \mid {e_1 \leq e_2}
\mid{b_1 \wedge b_2}
\mid {b_1 \vee b_2} \\
S\in \textit{Stmts} &  \Coloneqq & {x\coloneqq e}\mid {x\coloneqq \&
y}\mid {*x\coloneqq e}\mid {x\coloneqq *
y}\mid{\textit{skip}}\mid {S_1;S_2}\mid {\ifs} \mid \\
&  & {\while}\mid \textit{par}\{\{S_1\},\ldots,\{S_n\}\}\mid
\textit{par-if}\{(b_1,S_1),\ldots,(b_n,S_n)\}\mid
\textit{par-for}\{S\}.
\end{eqnarray*}
\caption{The programming language.}\label{lang}
\end{figure}

The parallel constructs include join-fork constructs, parallel
loops, and conditionally spawned threads. The \textit{par}
(join-fork) construct starts executing many concurrent threads at
the beginning of the \textit{par} construct and then waits until the
completion of all these executions at the end of the \textit{par}
construct. Semantically, the \textit{par} construct  can be
expressed approximately as if the threads are executed sequentially
in an arbitrary order. The parallel loop construct included in our
language is that of \textit{par-for}. This construct executes, in
parallel, a statically unknown number of threads each of which has
the same code (the loop body). Therefore the semantics of
\textit{par-for} can be expressed using that of the \textit{par}
construct. The construct including conditionally spawned threads is
that of \textit{par-if}. This construct executes, in parallel, its
$n$ concurrent threads. The execution of thread $(b_i,S_i)$ includes
the execution of $S_i$ only if $b_i$ is true.

One way to define the meaning of the constructs of our programming
language, including the parallel constructs, is by an operational
semantics. This amounts to defining a transition relation
$\rightsquigarrow$ between states which are defined as follows.
\begin{definition}
\begin{enumerate}
    \item $\textit{Addrs}=\{x^\prime\mid x\in \textit{Var}\}$ and
    ${\textit{Val}} = {\mathbb{Z} \cup \textit{Addrs}}.$
    \item \textit{state} $\in$ States = $\{\textit{abort}\}\cup
    \{\gamma\mid\gamma\in \Gamma= Var\longrightarrow Val\}$.
\end{enumerate}
\end{definition}

The semantics of arithmetic and Boolean expressions are defined as
usual except that arithmetic and Boolean operations are not allowed
on pointers.

{\footnotesize{
\[
{\lb n\rb \gamma} = {n}\quad {\lb \& x\rb \gamma} = {x^\prime}\quad
{\lb x\rb \gamma} = {\gamma(x)}\quad {\lb \textit{true}\rb \gamma} =
{\textit{true}}\quad {\lb \textit{false}\rb \gamma} =
{\textit{false}}
\]\[
{\lb *x\rb \gamma} = \left\{
\begin{array}{ll}
\gamma(y) &\mbox{if }{\gamma(x)}={y^\prime}, \\
! & \mbox{otherwise.}
\end{array} \right.
\quad {\lb e_1\oplus e_2\rb \gamma} = \left\{\begin{array}{ll} {\lb
e_1\rb \gamma}\oplus {\lb e_2\rb \gamma} & \mbox{if }{\lb e_1\rb
\gamma,\lb e_2\rb \gamma}\in{\mathbb{Z}}, \\
! &\mbox{otherwise.}
\end{array}
\right.
\]\[
{\lb \neg A\rb \gamma} = \left\{
\begin{array}{ll}
{\neg (\lb A\rb \gamma)} &\mbox{if }{\lb A\rb \gamma} \in
{\{\textit{true}, \textit{false}\},} \\! &\mbox{otherwise.}
\end{array} \right. \quad
{\lb e_1= e_2\rb \gamma} =\left\{\begin{array}{ll} ! & \mbox{if }
{\lb e_1\rb \gamma}= {!}
\mbox{ or } {\lb e_2\rb \gamma}= {!}, \\
\textit{true} & \mbox{if }{\lb e_1\rb \gamma}= { \lb e_2\rb
\gamma}\not = {!},\\ \textit{false} & \mbox{otherwise}.
\end{array}\right.
\]\[
{\lb e_1\le e_2\rb \gamma} =\left\{\begin{array}{ll} ! & \mbox{if }
{\lb e_1\rb \gamma}\not \in \mathbb{Z}
\mbox{ or } {\lb e_2\rb \gamma}\not \in \mathbb{Z}, \\
{\lb e_1\rb \gamma}\le { \lb e_2\rb \gamma} & \mbox{otherwise}.
\end{array}\right.
\]\[
\mbox{For }{\diamond} \in{\{\wedge,\vee\}},\  {\lb b_1\diamond
b_2\rb \gamma} =\left\{\begin{array}{ll} ! & \mbox{if } {\lb b_1\rb
\gamma}= {!}
\mbox{ or } {\lb b_2\rb \gamma}= {!}, \\
{\lb b_1\rb \gamma}\diamond { \lb b_2\rb \gamma} & \mbox{otherwise}.
\end{array}\right.
\]}}
The inference rules of our semantics (transition relation) are
defined as follows: {\footnotesize{
\[
\begin{prooftree}
{\lb e\rb\gamma} = {!}\justifies  x \coloneqq e:
\gamma\rightsquigarrow \textit{abort}\thickness=0.08em
\end{prooftree}\quad
\begin{prooftree}
{\lb e\rb\gamma} \not= {!}\justifies  x \coloneqq e:\gamma
\rightsquigarrow \gamma[x\mapsto\lb e\rb\gamma]\thickness=0.08em
\end{prooftree} \quad
\begin{prooftree} \gamma(x) =z^\prime \quad z \coloneqq e:\gamma
\rightsquigarrow \textit{state}\justifies *x \coloneqq e:
\gamma\rightsquigarrow \textit{state}\thickness=0.08em
\end{prooftree}
\]\[
\begin{prooftree} \gamma(x)\notin \textit{Addrs}
\justifies *x \coloneqq e: \gamma\rightsquigarrow
\textit{abort}\thickness=0.08em \end{prooftree}\quad
\begin{prooftree}
\justifies  x \coloneqq \& y:\gamma \rightsquigarrow \gamma[x\mapsto
y^\prime]\thickness=0.08em
\end{prooftree} \quad
\begin{prooftree} \gamma(y) =z^\prime \quad x \coloneqq z: \gamma\rightsquigarrow
\gamma^\prime \justifies x \coloneqq * y: \gamma\rightsquigarrow
\gamma^\prime \thickness=0.08em
\end{prooftree}\]\[
\begin{prooftree} \gamma(y)\notin \textit{Addrs}
 \justifies x \coloneqq *y: \gamma\rightsquigarrow
\textit{abort}\thickness=0.08em \end{prooftree} \quad
\begin{prooftree}
\justifies  \textit{skip}: \gamma \rightsquigarrow
\gamma\thickness=0.08em
\end{prooftree}
\quad
\begin{prooftree} S_1:\gamma \rightsquigarrow \textit{abort}\justifies
S_1;S_2: \gamma \rightsquigarrow \textit{abort}\thickness=0.08em
\end{prooftree}\quad \begin{prooftree}
S_1:\gamma \rightsquigarrow \gamma^{\prime\prime}\quad
S_2:\gamma^{\prime\prime}\rightsquigarrow state\justifies S_1;S_2:
\gamma \rightsquigarrow state\thickness=0.08em
\end{prooftree}
\]\[
\begin{prooftree}
{\lb b\rb \gamma}= {!} \justifies \ifs:\gamma \rightsquigarrow
\textit{abort}\thickness=0.08em
\end{prooftree}
\quad
\begin{prooftree}
{\lb b\rb \gamma}= {\textit{true}}  \quad S_t:\gamma
\rightsquigarrow \textit{state} \justifies \ifs:\gamma
\rightsquigarrow \textit{state}\thickness=0.08em
\end{prooftree}\]\[
\begin{prooftree}
{\lb b\rb \gamma}= {\textit{false}}  \quad S_f:\gamma
\rightsquigarrow \textit{state} \justifies \ifs:\gamma
\rightsquigarrow \textit{state}\thickness=0.08em
\end{prooftree}
\quad\begin{prooftree} {\lb b\rb \gamma}= {!} \justifies
\while:\gamma \rightsquigarrow \textit{abort}\thickness=0.08em
\end{prooftree}
\quad \begin{prooftree} {\lb b\rb \gamma}= {\textit{false}}
\justifies \while:\gamma \rightsquigarrow\gamma\thickness=0.08em
\end{prooftree}\]\[
\begin{prooftree}
{\lb b\rb \gamma}= {\textit{true}}  \quad S:\gamma \rightsquigarrow
\gamma^{\prime\prime} \quad \while:\gamma^{\prime\prime}
\rightsquigarrow \textit{state} \justifies \while:\gamma
\rightsquigarrow \textit{state}\thickness=0.08em
\end{prooftree} \quad\begin{prooftree}
{\lb b\rb \gamma}= {\textit{true}}  \quad S:\gamma \rightsquigarrow
\textit{abort} \justifies \while:\gamma \rightsquigarrow
\textit{abort}\thickness=0.08em
\end{prooftree}
\]}}
\begin{itemize}
\item[$\bullet$] \textbf{Join-fork}: {\footnotesize{\[
\begin{prooftree}
\justifies \textit{par}\{\{S_1\},\ldots,\{S_n\}\}:\gamma
\rightsquigarrow \gamma^\prime\using{\dagger}\thickness=0.08em
\end{prooftree}\quad \begin{prooftree}
\justifies \textit{par}\{\{S_1\},\ldots,\{S_n\}\}:\gamma
\rightsquigarrow \textit{abort}\using{\ddagger}\thickness=0.08em
\end{prooftree}\]
}}
\begin{itemize}
\item[$\dagger$] there exist a permutation
$\theta:\{1,\ldots,n\}\ra\{1,\ldots,n\}$ and $n+1$ states
$\gamma=\gamma_1,\ldots,\gamma_{n+1}=\gamma^\prime$ such that for
every $1\le i\le n,\ S_{\theta(i)}:\gamma_i\ra \gamma_{i+1}$.
\item[$\ddagger$] there exist $m$ such that $1\le m\le n$, a one-to-one map
$\beta:\{1,\ldots,m\}\ra\{1,\ldots,n\}$, and $m+1$ states
$\gamma=\gamma_1,\ldots,\gamma_{m+1}=\textit{abort}$ such that for
every $1\le i\le m,\ S_{\beta(i)}:\gamma_i\ra \gamma_{i+1}$.
\end{itemize}
\item[$\bullet$] \textbf{Conditionally spawned threads}: {\footnotesize{
\[\begin{prooftree}
\textit{par}\{\{\mathit{if\ b_1\ then\ }S_1 \mathit{\ else\
skip}\},\ldots,\{\mathit{if\ b_n\ then\ }S_n \mathit{\ else\
skip}\}\}:\gamma\rightsquigarrow \gamma^\prime\justifies
\textit{par-if}\{(b_1,S_1),\ldots,(b_n,S_n)\}:\gamma
\rightsquigarrow \gamma^\prime\thickness=0.08em
\end{prooftree}
\]\[
\begin{prooftree}
\textit{par}\{\{\mathit{if\ b_1\ then\ }S_1 \mathit{\ else\
skip}\},\ldots,\{\mathit{if\ b_n\ then\ }S_n \mathit{\ else\
skip}\}\}:\gamma\rightsquigarrow \textit{abort}\justifies
\textit{par-if}\{(b_1,S_1),\ldots,(b_n,S_n)\}:\gamma
\rightsquigarrow \textit{abort}\thickness=0.08em
\end{prooftree}
\]}}\item[$\bullet$] \textbf{Parallel loops}: {\footnotesize{\[
\begin{prooftree}
\exists n.\
\textit{par}\{\overbrace{\{S\},\ldots,\{S\}}^{n-times}\}:\gamma\rightsquigarrow
\gamma^\prime\justifies \textit{par-for}\{S\}:\gamma
\rightsquigarrow \gamma^\prime\thickness=0.08em
\end{prooftree}\qquad\begin{prooftree}
\exists n.\ \textit{par}\{\overbrace{\{S\},\ldots,\{S\}}^{n-times}
\}:\gamma\rightsquigarrow \textit{abort}\justifies
\textit{par-for}\{S\}:\gamma \rightsquigarrow
\textit{abort}\thickness=0.08em
\end{prooftree}\]
}} \end{itemize}

A simple example for the \textit{par-for} command is
$\textit{par-for}\{x:=10\}$. The execution of this command amounts
to fixing a number randomly, say $7$, and then to concurrently
execute the seven threads $\{x:=10\}_1,\ldots,\{x:=10\}_7$.
\section{Pointer analysis}\label{pointers}
In this section, we present a novel technique for flow-sensitive
pointer analysis of structured parallel programs where shared
pointers may be updated simultaneously. Our technique manipulates
important parallel constructs; join-fork constructs, parallel loops,
and conditionally spawned threads. The proposed technique has the
form of a compositional type system which is simply structured.
Consequently results of the analysis are in the form of types
assigned to expressions and statements approved by type derivations.
Therefore a type is assigned to each program point of a statement
(program). This assigned type specifies for each variable in the
program a conservative approximation of the addresses that may be
assigned to the variable. The set of points-to types \textit{PTS}
and the relation ${\models}\subseteq {\Gamma\times \textit{PTS}}$
are defined as follows:\begin{definition}
\begin{enumerate}
\item $\mathit{PTS} = \{ \pts\mid \pts:\textit{Var}
\ra 2^{\textit{Addrs}}\}$.
\item $\pts \le \ptsp\ \df\ \forall x\in\textit{Var}.\
\pts(x)\subseteq \ptsp(x)$.
\item $\gamma\models\pts \ \df\ (\forall x\in \textit{Var}.\
\gamma(x)\in \textit{Addrs} \Ra \gamma(x)\in \pts(x))$.
\end{enumerate}
\end{definition}

The judgement of an expression has the form $e:\pts\ra A$. The
intended meaning of this judgment, which is formalized in
Lemma~\ref{lem1}, is that $A$ is the collection of addresses that
$e$ may evaluate to in a state of type $\pts$. The judgement of a
statement has the form $S:\pts\ra \ptsp$. This judgement simply
guarantees that if $S$ is executed in a state of type $\pts$ and the
execution terminates in a state $\gamma^\prime$, then
$\gamma^\prime$ has type $\ptsp$. Typically, the pointer analysis
for a program $S$ is achieved via a post-type derivation for the
bottom type (mapping variables to $\emptyset$) as the pre-type.

 The inference rules of our type system for pointer analysis are the
following: {\footnotesize {\[
\begin{prooftree} \justifies n:
\pts\ra \emptyset \thickness=0.08em
\end{prooftree} \quad
\begin{prooftree}
\justifies x: \pts\ra \pts(x) \thickness=0.08em
\end{prooftree}\quad \begin{prooftree}
\justifies e_1\oplus e_2: \pts\ra\emptyset \thickness=0.08em
\end{prooftree}
\quad
\begin{prooftree}
e:\pts\ra A\justifies x \coloneqq e: \pts\ra \pts[x\mapsto
A]\thickness=0.08em\using{(\coloneqq^p)}
\end{prooftree}\]\[ \begin{prooftree}
\justifies x \coloneqq \& y: \pts\ra \pts[x\mapsto
\{y^\prime\}]\thickness=0.08em \using{(\coloneqq
\&^p)}\end{prooftree} \qquad
\begin{prooftree}
\justifies \textit{skip}: \pts\ra \pts\thickness=0.08em
\end{prooftree}
\]\[
\begin{prooftree} \forall z^\prime \in \pts(y).\ x\coloneqq z:
\pts\ra\ptsp\justifies x \coloneqq *y: \pts\ra
\ptsp\thickness=0.08em \using{(\coloneqq *^p)}
\end{prooftree}\quad\begin{prooftree}
\forall z^\prime \in \pts(x).\ z\coloneqq e: \pts\ra\ptsp \justifies
*x \coloneqq e: \pts\ra \ptsp\thickness=0.08em\using{(*\coloneqq^p)}
\end{prooftree}
\]\[
\begin{prooftree}
S_i:\pts\cup \cup_{j\not=i} \pts_j\ra \pts_i\justifies
\textit{par}\{\{S_1\},\ldots,\{S_n\}\}: \pts\ra \cup_i
\pts_i\thickness=0.08em
\using{(\textit{par}^p)}\end{prooftree}\qquad
\begin{prooftree} S_1: \pts \ra \pts^{\prime\prime}
\quad S_2: \pts^{\prime\prime} \ra \ptsp\justifies S_1;S_2: \pts\ra
\ptsp \thickness=0.08em\using{(\textit{seq}^p)}
\end{prooftree}
\]\[\begin{prooftree}
\textit{par}\{\{\mathit{if\ b_1\ then\ }S_1 \mathit{\ else\
skip}\},\ldots,\{\mathit{if\ b_n\ then\ }S_n \mathit{\ else\
skip}\}\}: \pts\ra \ptsp\justifies
\textit{par-if}\{(b_1,S_1),\ldots,(b_n,S_n)\}: \pts\ra
\ptsp\thickness=0.08em \using{(\textit{par-if}^p)}\end{prooftree}
\]\[
\begin{prooftree}
S: \pts\cup \ptsp\ra \ptsp\justifies \textit{par-for}\{S\}: \pts\ra
\ptsp\thickness=0.08em \using{(\textit{par-for}^p)}\end{prooftree}
\qquad
\begin{prooftree}
S_t:\pts\ra \ptsp \quad S_f:\pts\ra \ptsp\justifies \ifs: \pts\ra
\ptsp \thickness=0.08em\using{(\textit{if}^p)}
\end{prooftree}\]\[
\begin{prooftree}
S_t: \pts\ra \pts \justifies \while: \pts\ra \pts
\thickness=0.08em\using{(\textit{whl}^p)}
\end{prooftree}\qquad
\begin{prooftree}
\ptsp_1\leq \pts_1\quad S:\pts_1 \ra \pts_2 \quad \pts_2 \leq
\ptsp_2\justifies S: \ptsp_1\ra
\ptsp_2\thickness=0.08em\using{(\textit{csq}^p)}
\end{prooftree}
\]
}}

The inference rules corresponding to assignment commands are clear.
For the rule $(\textit{par}^p)$ of the join-fork command,
\textit{par}, one possibility is that the execution of a thread
$S_i$ starts before the execution of any other thread starts.
Another possibility is that the execution starts after executions of
all other threads end. Of course there are many other possibilities
in between. Consequently, the analysis of the thread $S_i$ must
consider all such possibilities. This is reflected in the pre-type
of $S_i$ and the post-type of the \textit{par} command. Clearly
$pts$, in $(\textit{par}^p)$, is the given pre-type for the
\textit{par} command and for witch the rule calculates a post-type.
The union operation in the rule $(\textit{par}^p)$ makes the
pre-type of threads general enough. Similar explanations clarify the
rules $(par-if^p)$ and $(par-for^p)$.

We note that a type invariant is required to type a \textit{while}
statement. Also to achieve the analysis for one of the
\textit{par}'s threads we need to know the analysis results for all
other threads. However obtaining these results requires the result
of analyzing the first thread. Therefore there is a kind of
circularity in rule $(\textit{par}^p)$. Similar situations are in
rules $(\textit{par-if}^p)$ and $(par-for^p)$. Such issues can be
treated using a fix-point algorithm. The convergence of this
algorithm is guaranteed as the rules of our type system are monotone
and the set of points-to types \textit{PTS} is a complete lattice.
What makes calculations actual simple is that for any given program
the lattice \textit{PTS} is finite. The rule $(\textit{csq}^p)$ is
necessary to calculate a type invariant.
\begin{lemma}\label{lem1}
\begin{enumerate}
\item Suppose $e:\pts \ra A$ and $\gamma\models \pts$. Then $\lb e\rb
\gamma \in \textit{Addrs}$ implies $\lb e\rb \gamma \in A$.
\item $\pts\le \pts^\prime \Longleftrightarrow (\forall \gamma.\ \gamma\models
\pts\Ra \gamma\models \pts^\prime)$.
\end{enumerate}
\end{lemma}
\begin{proof}
The first item is obvious. The left-to-right direction of $(2)$ is
easy. The other direction is proved as follows. Suppose
$y^\prime\in\pts(x)$. Then the state $\{(x,y^\prime),(t,0)\mid t\in
\textit{Var}\setminus\{x\}\}$ is of type $\pts$ and hence of type
$\ptsp$ implying that $y^\prime\in\ptsp(x)$. Therefore
$\pts(x)\subseteq \ptsp(x)$. Since $x$ is arbitrary, $\pts\leq
\ptsp$.
\end{proof}

\begin{theorem} \( (Soundness)\)\label{soundness-points to}
Suppose that $S:\pts\ra \ptsp,\ S:\gamma\rightsquigarrow
\gamma^\prime$, and ${\gamma}\models {pts}$. Then ${\gamma^\prime}
\models {\ptsp}$.
\end{theorem}
\begin{proof}
The proof is by structural induction on the type derivation. We
demonstrate some cases.
\begin{itemize}
\item The case of $(\coloneqq^p)$: In this case $\ptsp=pts[x\mapsto A]$
and $\gamma^\prime=\gamma[x\mapsto \lb e\rb\gamma]$. Therefore by
the previous lemma $\gamma \models \pts$ implies $\gamma^\prime
\models \ptsp$.
\item The case of $(*\coloneqq^p)$: In this case there exists
$z\in \textit{Var}$ such that $\gamma(x)=z^\prime$ and $z\coloneqq
e:\gamma\rightsquigarrow \gamma^\prime$. Because $\gamma \models
\pts,\ z^\prime\in \pts(x)$ and hence by assumption $z\coloneqq
e:\pts\ra \ptsp$. Therefore by soundness of $(\coloneqq^p)$,
$\gamma^\prime\models \ptsp$.
\item The case of $(par^p)$: In this case there exist a permutation
$\theta:\{1,\ldots,n\}\ra\{1,\ldots,n\}$ and $n+1$ states
$\gamma=\gamma_1,\ldots,\gamma_{n+1}=\gamma^\prime$ such that for
every ${1\le i\le n},\ S_{\theta(i)}:\gamma_i\ra \gamma_{i+1}$. Also
$\gamma_1\models \pts$ implies $\gamma_1\models \pts \cup
\cup_{j\not=\theta(1)} pts_j$. Therefore by the induction hypothesis
$\gamma_2\models \pts_{\theta(1)}$. This implies $\gamma_2\models
\pts \cup \cup_{j\not= \theta(2)}\pts_j$. Again by the induction
hypothesis we get $\gamma_3\models \pts_{\theta(2)}$. Therefore by a
simple induction on $n$, we can show that
$\gamma^\prime=\gamma_{n+1}\models \pts_{\theta(n)}$ which implies
${\gamma^\prime\models \pts^\prime=\cup_j\pts_j}$.
\item The case of $(par-for^p)$: In this case there exists $n$ such that $
\textit{par}\{\overbrace{\{S\},\ldots,\{S\}}^{n-times}\}:\gamma\rightsquigarrow
\gamma^\prime$. By induction hypothesis we have $S:\pts\cup\ptsp\ra
\ptsp$. By $(par^p)$ we conclude that
$\textit{par}\{\overbrace{\{S\},\ldots,\{S\}}^{n-times}\}:\pts\rightsquigarrow
\ptsp$. Therefore by the soundness of $(par^p)$,
$\gamma^\prime\models \ptsp$.
\end{itemize}
\end{proof}

\section{Live-variables analysis}\label{live}
In this section, we present a type system to perform live-variables
analysis for pointer programs with structured parallel constructs.
We start with defining live-variables:
\begin{definition}\label{livedef}
\begin{itemize}
\item A variable is \emph{usefully used} if it is used
\begin{itemize}
    \item as the operand of the unary operation $*$.
    \item in an assignment to a variable that is live at the end of the
    assignment, or
    \item in the guard of an if-statement or a while-statement,
\end{itemize}
\item A variable is \emph{live} at a program point if there
is a computational path from that program point during which the
variable gets usefully used before being modified.
\end{itemize}
\end{definition}
\begin{definition}\label{live types}
The set of \emph{live types} is denoted by $L$ and equal to
$\textit{PTS}\times\mathcal{P}(\textit{Var})$. The second component
of a live type is termed a \textit{live-component}. The subtyping
relation $\leq$ is defined as: ${(\pts,l)\leq(\ptsp,l^\prime)\ \df\
\pts\leq \ptsp \mbox{ and } l\supseteq l^\prime.}$
\end{definition}

The live-variables analysis is a backward analysis. For each program
point, this analysis specifies the set of variables that may be live
(according to the definition above) at that point.

Our type system for live-variables analysis is obtained as an
enrichment of the type system for pointer analysis, presented in the
previous section. Hence one can say that the type system presented
here is a strict extension of that presented above. This is so
because the result of pointer analysis is necessary to improve the
precision of the live-variables analysis. This also gives an
intuitive explanation of the definition of live types above.

The judgement of a statement $S$ has the form $S:(\pts,l)\ra
(\ptsp,l^\prime)$. The intuition of the judgement is that the
presence of live-variables at the post-state of an execution of $S$
in $l^\prime$ implies the presence of live-variables at the
pre-state of this execution in $l$. The intuition agrees with the
fact that live-variables analysis is a backward analysis and gives
an insight into the definition of $\gamma \models l$ below.

Suppose we have the set of variables $l^\prime $ that we have
interest in their values at the end of executing a statement $S$ and
the result of pointer analysis of $S$ (in the form $S:\pts\ra
\ptsp$). The live-variables analysis takes the form of a pre-type
derivation that calculates a set $l$ such that $S:(\pts,l)\ra
(\ptsp,l^\prime)$. The idea here is that by proceeding from the last
point of the program, calculating pre-types, we achieve the backward
analysis.

The inference rules for our type system for live-variables analysis
are as follows. {\footnotesize {
\[
\begin{prooftree}
x \coloneqq e: \pts\ra \ptsp \quad x\notin l^\prime\justifies x
\coloneqq e: (\pts,l^\prime)\ra
(\ptsp,l^\prime)\thickness=0.08em\using{(\coloneqq^l_1)}
\end{prooftree}\quad \begin{prooftree}
x \coloneqq e: \pts\ra \ptsp \quad x\in l^\prime\justifies x
\coloneqq e: (\pts,(l^\prime\setminus \{x\}) \cup FV(e))\ra
(\ptsp,l^\prime)\thickness=0.08em\using{(\coloneqq^l_2)}
\end{prooftree}
\]\[
\begin{prooftree}
\justifies x \coloneqq \& y: (\pts,l^\prime\setminus\{x\})\ra
(\pts[x\mapsto \{y^\prime\}],l^\prime)\thickness=0.08em
\using{(\coloneqq \&^l)}\end{prooftree} \quad
\begin{prooftree}
\justifies \textit{skip}: (\pts,l)\ra (\pts,l)\thickness=0.08em
\end{prooftree}
\]\[
\begin{prooftree} x \coloneqq *y: \pts\ra
\ptsp\quad x\notin l^\prime\justifies x \coloneqq *y:
(\pts,l^\prime\cup\{y\})\ra (\ptsp,l^\prime)\thickness=0.08em
\using{(\coloneqq *^l_1)}
\end{prooftree}\]\[ \begin{prooftree} x \coloneqq *y: \pts\ra
\ptsp\quad x\in l^\prime\justifies x \coloneqq *y:
(\pts,(l^\prime\setminus\{x\})\cup\{y,z\mid z^\prime\in
\pts(y)\})\ra (\ptsp,l^\prime)\thickness=0.08em \using{(\coloneqq
*^l_2)}
\end{prooftree}
\]\[
\begin{prooftree}
*x \coloneqq e: \pts\ra \ptsp \quad {\pts(x)\cap
l^\prime}={\emptyset} \justifies *x \coloneqq e:
(\pts,l^\prime\cup\{x\})\ra
(\ptsp,l^\prime)\thickness=0.08em\using{(*\coloneqq^l_1)}
\end{prooftree}\quad \begin{prooftree}
*x \coloneqq e: \pts\ra \ptsp \quad {\pts(x)\cap
l^\prime}\not={\emptyset} \justifies *x \coloneqq e:
(\pts,l^\prime\cup \textit{FV}(e)\cup\{x\})\ra
(\ptsp,l^\prime)\thickness=0.08em\using{(*\coloneqq^l_2)}
\end{prooftree}
\]\[
\begin{prooftree}
S_i:(\pts\cup \cup_{j\not=i} \pts_j,l_i)\ra (\pts_i,l^\prime\cup
\cup_{j\not=i} l_j)\justifies
\textit{par}\{\{S_1\},\ldots,\{S_n\}\}: (\pts,\cup_i l_i)\ra (\cup_i
\pts_i,l^\prime)\thickness=0.08em
\using{(\textit{par}^l)}\end{prooftree}
\]\[
\begin{prooftree} \textit{par}\{\{\mathit{if\ b_1\ then\ }S_1
\mathit{\ else\ skip}\},\ldots,\{\mathit{if\ b_n\ then\ }S_n
\mathit{\ else\ skip}\}\}: (\pts,l)\ra (\ptsp,l^\prime)\justifies
\textit{par-if}\{(b_1,S_1),\ldots,(b_n,S_n)\}: (\pts,l)\ra
(\ptsp,l^\prime)\thickness=0.08em
\using{(\textit{par-if}^l)}\end{prooftree}
\]\[
\begin{prooftree}
S: (\pts\cup \ptsp,l)\ra (\ptsp,l^\prime\cup l)\justifies
\textit{par-for}\{S\}: (\pts,l)\ra (\ptsp,l^\prime)\thickness=0.08em
\using{(\textit{par-for}^l)}\end{prooftree}
\]\[
\begin{prooftree} S_1: (\pts,l) \ra (\pts^{\prime\prime},l^{\prime\prime})
\quad S_2: (\pts^{\prime\prime},l^{\prime\prime}) \ra
(\ptsp,l^\prime)\justifies S_1;S_2: (\pts,l)\ra (\ptsp,l^\prime)
\thickness=0.08em\using{(seq^l)}
\end{prooftree}
\]\[
\begin{prooftree}
S_t,S_f:(\pts,l)\ra (\ptsp,l^\prime)\justifies \ifs: (\pts,l\cup
\textit{FV}(b))\ra (\ptsp,l^\prime)
\thickness=0.08em\using{(\textit{if}^l)}
\end{prooftree}
\]\[
\begin{prooftree}
{l}={l^\prime\cup \textit{FV}(b)}\quad S_t: (\pts,l^\prime)\ra
(\pts,l) \justifies \while: (\pts,l)\ra (\pts,l^\prime)
\thickness=0.08em\using{(\textit{whl}^l)}
\end{prooftree}
\]\[
\begin{prooftree}
(\ptsp_1,l_1^\prime)\leq (\pts_1,l_1)\quad S:(\pts_1,l_1) \ra
(\pts_2,l_2) \quad (\pts_2,l_2) \leq (\ptsp_2,l_2^\prime)\justifies
S: (\ptsp_1,l_1^\prime)\ra
(\ptsp_2,l_2^\prime)\thickness=0.08em\using{(\textit{csq}^l)}
\end{prooftree}
\]
}}

For the command $*x \coloneqq e$, we have two rules, namely
$(*\coloneqq^l_1)$ and $(*\coloneqq^l_2)$. In both cases,
calculating the pre-type from the post-type includes adding $x$ to
the post-type. This is so because according to
Definition~\ref{livedef}, $x$ is live at the pre-state of any
execution of the command. The rule $(*\coloneqq^l_1)$ deals with the
case that there is no possibility that the variable modified by this
statement is live ($\pts(x)\cap l^\prime = \emptyset$) at the end of
an execution. In this case there is no need to add any other
variables to the post-type. The rule $(*\coloneqq^l_2)$ deals with
the case that there is a possibility that the variable modified by
this statement is live ($\pts(x)\cap l^\prime \not= \emptyset$) at
the end of an execution. In this case, there is a possibility that
free variables of $e$ are used usefully according to
Definition~\ref{livedef}. Therefore free variables of $e$ are added
to the post-type. This gives an intuitive explanation for rules of
all the assignment commands. The intuition given in the previous
section for the rules $(\textit{par}^p)$ helps to understand the
rules for the parallel constructs, $(\textit{par}^l),\
(\textit{par-if}^l)$, and $(\textit{par-for}^l)$.

Towards proving the soundness of our type system for live-variables
analysis, we introduce necessary definitions and results.
\begin{definition}\label{statequi}
\begin{enumerate}
\item $\gamma\models_{l} \pts\ \df\
\forall x\in l.\ \gamma(x)\in \textit{Addrs} \Ra \gamma(x)\in
\pts(x).$
\item $\gamma\thicksim_{l} \gamma^\prime\ \df\
\forall x \in l.\ \gamma(x)=\gamma^\prime(x).$
\item $\gamma\thicksim_{(\pts,l)} \gamma^\prime\
  \df\ \gamma\models_{l} \pts,
  \gamma^\prime\models_{l} \pts$, and $ \gamma\thicksim_{l}
  \gamma^\prime.$
\end{enumerate}
\end{definition}
\begin{definition}\label{statetype}
The expression $\gamma \models l$ denotes the case when there is a
variable that is live at that state (computational point) and is not
included in $l$. A state $\gamma$ has type $(\pts,l)$, denoted by
$\gamma \models (\pts,l)$, if $\gamma\models_{l} \pts$ and
$\gamma\models l$.
\end{definition}

The following lemma is proved by structural induction on $e$ and
$b$.
\begin{lemma}
Suppose that $\gamma \mbox{ and }\gamma^\prime$ are states and $l
\mbox{ and }l^\prime\in\mathcal{P}(\textit{Var})$. Then
\begin{enumerate}
\item If $l\supseteq l^\prime$ and $\gamma
\thicksim_{l}\gamma^\prime$, then
$\gamma\thicksim_{l^\prime}\gamma^\prime$.

\item If ${l} ={l^\prime\cup FV(e)}\mbox{ and }\gamma
\thicksim_{l}\gamma^\prime$, then ${\lb e\rb \gamma} = {\lb e\rb
\gamma^\prime}$ and ${\gamma \thicksim_{l^\prime}\gamma^\prime}$.

\item If ${l} ={l^\prime\cup FV(b)} \mbox{ and }\gamma
\thicksim_{l}\gamma^\prime$, then ${\lb b\rb \gamma} = {\lb b\rb
\gamma^\prime}$ and ${\gamma \thicksim_{l^\prime}\gamma^\prime}$.
\end{enumerate}
\end{lemma}

The following lemma follows from Lemma~\ref{lem1}.
\begin{lemma}\label{lem2}
Suppose that $\gamma\models_{l} \pts, FV(e)\subseteq l, \mbox{and
}e:\pts \ra A$. Then
\[\lb e\rb \gamma\in \textit{Addrs} \Ra
\lb e\rb \gamma\in A.\]
\end{lemma}
\begin{proof}
Consider the state $\gamma^\prime$, where $\gamma^\prime= {\lambda
x.\mbox{ if } x\in FV(e) \mbox{ then } \gamma(x) \mbox{ else } 0}$.
It is not hard to see that $\lb e\rb \gamma = \lb e\rb
\gamma^\prime$ and $\gamma^\prime\models \pts$. Now by
Lemma~\ref{lem1}, $\lb e\rb \gamma^\prime \in \textit{Addrs}$
implies $\lb e\rb \gamma^\prime \in A$ which completes the proof.
\end{proof}

\begin{theorem} \label{livesoundness1}
\begin{enumerate}
    \item  $(\pts,l)\le (\ptsp,l^\prime) \Ra
    (\forall \gamma.\ {\gamma}\models_{l} {\pts} \Ra
    {\gamma}\models_{l^\prime} {\ptsp}$).
    \item  Suppose that $S:(\pts,l)\ra
    (\ptsp,l^\prime)$ and
$S:\gamma\rightsquigarrow \gamma^\prime$. Then ${\gamma}\models_{l}
{\pts}$ implies ${\gamma^\prime} \models_{l^\prime} {\ptsp}.$
    \item Suppose that $S:(\pts,l)\ra (\ptsp,l^\prime)$ and
$S:\gamma\rightsquigarrow \gamma^\prime$. Then $\gamma\models l$
implies $\gamma^\prime\models l^\prime$. This guarantees that if the
set of variables live at $\gamma^\prime$ is included in $l^\prime$,
then the set of variables live at $\gamma$ is included in $l$.
\end{enumerate}
\end{theorem}
\begin{proof}
\begin{enumerate}
    \item Suppose ${\gamma}\models_{l} {\pts}$.
    This implies ${\gamma}\models_{l^\prime} {\pts}$
    because $l^\prime \subseteq l$. The last
    fact implies ${\gamma}\models_{l^\prime} {\ptsp}$
    because $\pts\le \ptsp$.
    \item The proof is by induction on the structure of
    type derivation. We show some cases.
\begin{enumerate}
\item The type derivation has the form $(\coloneqq^l_1)$.
In this case, $\ptsp = \pts[x\mapsto A]$ and
${\gamma^\prime}={\gamma[x\mapsto \lb e\rb \gamma]}$. Therefore
${\gamma}\models_{l^\prime} {\pts}$ implies
${\gamma^\prime}\models_{l^\prime} {\ptsp}$ because $x\notin
l^\prime$.
\item The type derivation has the form $(\coloneqq^l_2)$.
In this case, $e:\pts\ra A,\ \ptsp = {\pts[ x\mapsto A]},\
\gamma^\prime=\gamma[x \mapsto\lb e\rb \gamma]$, and
$l=(l^\prime\setminus \{x\}) \cup FV(e)$. Therefore by
Lemma~\ref{lem2} it is not hard to see
${\gamma^\prime}\models_{l^\prime} {\ptsp}$.
\item The type derivation has the form $(\coloneqq *^l_1)$.
In this case, for every $z^\prime \in \pts(y)$, we have $x\coloneqq
z: \pts\ra\ptsp, \gamma(y)=z^\prime$, and $x\coloneqq z:\gamma \ra
\gamma^\prime$. We have $z^\prime\in \pts(y)$, because $y\in l$ and
$\gamma\models_l\pts$. Therefore by $(\coloneqq^l_1)$, we have
$x\coloneqq z: (\pts,l^\prime)\ra(\ptsp,l^\prime)$. Now
$\gamma\models_l\pts$ amounts to $\gamma\models_{l^\prime}\pts$.
Hence we get $\gamma^\prime\models_{l^\prime}\ptsp$ by soundness of
$(\coloneqq^l_1)$.
\item The type derivation has the form $(\coloneqq *^l_2)$.
In this case, for every $z^\prime \in \pts(y)$, we have $ x\coloneqq
z: \pts\ra\ptsp, \gamma(y)=z^\prime,\ x\coloneqq z:\gamma \ra
\gamma^\prime,$ and $l=(l^\prime\setminus \{x\}) \cup \{y,z\mid
z^\prime\in \pts(y)\}$. We have $z\in\pts(y)$ because
$\gamma\models_l \pts$ and $y\in l$. Therefore by $(\coloneqq^l_2)$
we have $x\coloneqq z:
(\pts,(l^\prime\setminus\{x\})\cup\{z\})\ra(\ptsp,l^\prime)$.
$\gamma\models_l \pts$ implies
$\gamma\models_{(l^\prime\setminus\{x\})\cup\{z\}} \pts$. Hence by
soundness of $(\coloneqq^l_2)$, we get
$\gamma^\prime\models_{l^\prime}\ptsp$.
\item The type derivation has the form $(*\coloneqq^l_1)$.
In this case, for every $z^\prime \in \pts(x)$, we have $z\coloneqq
e: \pts\ra\ptsp, \gamma(x)=z^\prime$, and $z\coloneqq e:\gamma \ra
\gamma^\prime$. We have $z^\prime\in \pts(x)$, because $x\in l$ and
$\gamma\models_l\pts$. Therefore by $(\coloneqq^l_1)$, we have
$z\coloneqq e: (\pts,l^\prime)\ra(\ptsp,l^\prime)$ because
$\pts(x)\cap l^\prime =\emptyset$. Now $\gamma\models_l\pts$ amounts
to $\gamma\models_{l^\prime}\pts$. Hence we get
$\gamma^\prime\models_{l^\prime}\ptsp$ because $z\coloneqq e:
(\pts,l^\prime)\ra(\ptsp,l^\prime)$ and by soundness of
$(\coloneqq^l_1)$.
\item The type derivation has the form $(*\coloneqq^l_2)$.
In this case, for every $z^\prime \in \pts(x)$, we have $z\coloneqq
e: \pts\ra\ptsp, \gamma(x)=z^\prime,\ z\coloneqq e:\gamma \ra
\gamma^\prime,$ and $l=l^\prime \cup FV(e)\cup\{x\}$. We have
$z\in\pts(x)$ because $\gamma\models_l \pts$ and $x\in l$. Therefore
by $(\coloneqq^l_2)$ we have $x\coloneqq z:
(\pts,(l^\prime\setminus\{z\})\cup FV(e))\ra(\ptsp,l^\prime)$.
$\gamma\models_l \pts$ implies
$\gamma\models_{(l^\prime\setminus\{z\})\cup FV(e)} \pts$. Hence by
soundness of $(\coloneqq^l_2)$, we get
$\gamma^\prime\models_l^\prime\ptsp$.
\item The type derivation has the form $(\textit{par}^l)$.
In this case there exist a permutation
$\theta:\{1,\ldots,n\}\ra\{1,\ldots,n\}$ and $n+1$ states
$\gamma=\gamma_1,\ldots,\gamma_{n+1}=\gamma^\prime$ such that for
every $1\le i\le n,\ S_{\theta(i)}:\gamma_i\ra \gamma_{i+1}$. Also
$\gamma_1\models_l \pts$ implies $\gamma_1\models_{l_{\theta(1)}}
\pts \cup \cup_{j\not=\theta(1)} pts_j$. Therefore by the induction
hypothesis $\gamma_2\models_{l^\prime\cup \cup_{j\not=\theta(1)}
l_j} \pts_{\theta(1)}$. This implies
$\gamma_2\models_{l_{\theta(2)}} \pts \cup \cup_{j\not=
\theta(2)}\pts_j$. Again by the induction hypothesis we get
$\gamma_3\models_{l^\prime\cup \cup_{j\not=\theta(2)}l_j}
\pts_{\theta(2)}$. Therefore by a simple induction on $n$, we can
show that $\gamma^\prime=\gamma_{n+1}\models_{l^\prime\cup
\cup_{j\not=\theta(n)}l_j} \pts_{\theta(n)}$ which implies
${\gamma^\prime\models_{l^\prime} \pts^\prime=\cup_j\pts_j}$.
\item The type derivation has the form $(\textit{par-for}^l)$:
In this case there exists $n$ such that $
\textit{par}\{\overbrace{\{S\},\ldots,\{S\}}^{n-times}\}:\gamma\rightsquigarrow
\gamma^\prime$. By induction hypothesis we have
$S:(\pts\cup\ptsp,l)\ra (\ptsp,l\cup l^\prime)$. By $(par^l)$ we
conclude that
$\textit{par}\{\overbrace{\{S\},\ldots,\{S\}}^{n-times}\}:(\pts,l)\ra
(\ptsp,l^\prime)$. Therefore by soundness of $(par^l)$, we get
$\gamma^\prime\models_{l^\prime} \ptsp$.
\end{enumerate}
\item The proof is also by induction on the structure of
type derivation and it is straightforward.
\end{enumerate}
\end{proof}

The proof of the following corollary follows from
Theorem~\ref{livesoundness1}.
\begin{corollary}
Suppose $S:\gamma\rightsquigarrow \gamma^\prime\mbox{ and }
S:(\pts,l) \ra (\ptsp,l^\prime)$. Then $\gamma\models (\pts,l)$
implies $\gamma^\prime\models (\ptsp,l^\prime)$.
\end{corollary}

\begin{theorem} \label{soundness-live}
Suppose that $S:(\pts,l) \ra (\ptsp,l^\prime),\
S:\gamma\rightsquigarrow \gamma^\prime,\
\gamma\thicksim_{(\pts,l)}\gamma_*,$ and $S$ does not abort at
$\gamma_*$. Then there exists a state $\gamma_*^\prime$ such that
$S:\gamma_*\ra\gamma_*^\prime$ and $\gamma^\prime
\thicksim_{(\ptsp,l^\prime)}\gamma_*^\prime$.
\end{theorem}
\begin{proof}
The proof is by induction on structure of type derivation. We
demonstrate some cases:
\begin{enumerate}
\item The type derivation has one of the forms
$(\coloneqq^l_1)$ and $(\coloneqq^l_2)$. In this case, $\ptsp =
\pts[x\mapsto A]$ and $\gamma^\prime=\gamma[x\mapsto\lb e\rb
\gamma]$. We take $\gamma_*^\prime=\gamma_*[x\mapsto\lb e\rb
\gamma_*]$.
\item The type derivation has the form $(\coloneqq *^l_1)$ or
$(\coloneqq *^l_2)$. In this case, $\forall z^\prime \in \pts(y)$,
we have $x\coloneqq z: \pts\ra\ptsp, \gamma(y)=z^\prime$, and
$x\coloneqq z:\gamma \ra \gamma^\prime$. We set
${\gamma_*^\prime}={\gamma_*[x\mapsto \gamma_*(z)]}$.
\item The type derivation has one of the forms $(\coloneqq *^l_1)$ and
$(\coloneqq *^l_2)$. In this case, $\forall z^\prime \in \pts(x)$,
we have $z\coloneqq e: \pts\ra\ptsp, \gamma(x)=z^\prime$, and
$z\coloneqq e:\gamma \ra \gamma^\prime$. We let
${\gamma_*^\prime}={\gamma_*[z\mapsto\lb e\rb \gamma_*]}$
\item The type derivation has the form $(\textit{par}^l)$.
In this case there exist a permutation
$\theta:\{1,\ldots,n\}\ra\{1,\ldots,n\}$ and $n+1$ states
$\gamma=\gamma_1,\ldots,\gamma_{n+1}=\gamma^\prime$ such that for
every $1\le i\le n,\ S_{\theta(i)}:\gamma_i\ra \gamma_{i+1}$. We
refer to $\gamma_*$ as $\gamma_{*1}$.  We have
$\gamma_1\thicksim_{(\pts,\cup_i l_i)}\gamma_{*1}$ which implies
$\gamma_1\thicksim_{(\pts\cup\cup_{j\not=\theta(1)}\pts_j,l_{\theta(1)})}\gamma_{*1}$.
Therefore by induction hypothesis, there exists $\gamma_{*2}$ such
that $S_{\theta(1)}:\gamma_{*1}\ra\gamma_{*2}$ and
$\gamma_2\thicksim_{(\pts_{\theta(1)},l^\prime\cup \cup_j l_{j\not
=\theta(1)})}\gamma_{*2}$ which implies
$\gamma_2\thicksim_{(\pts\cup\cup_{j\not=\theta(2)}\pts_j,l_{\theta(2)})}\gamma_{*2}$.
Therefore a simple induction on $n$ proves the required.
\end{enumerate}
\end{proof}

\section{Dead code elimination}\label{optimization}
This section introduces a type system for dead code elimination.
Given a program and a set of variables whose values concern us at
the end of the program, there may be some code in the program that
has no effect on the values of these variables. Such code is called
\textit{dead code}.  The type system presented here aims at
optimizing structured parallel programs with pointer constructs via
eliminating dead code. In the form of a type derivation, the type
system associates each optimization with a proof for the soundness
of the optimization. Optimizing a program may result in correcting
it i.e. preventing it from aborting. Of course this happens if the
removed dead code is the only cause of abortion.

The type system presented here has judgements of the the form: $S:
(\pts,l)\ra (\ptsp,l^\prime)\hookrightarrow S^\prime.$ The intuition
is that $S^\prime$ optimizes $S$ towards dead code elimination (and
may be program correction). As mentioned early in many occasions,
the derivation of such judgement provides a justification for the
optimization process. The form of the judgement makes it apparent
that the type system presented in this section is built on the type
system for live-variables analysis.
\begin{figure}[htb]
{\footnotesize{\begin{itemize}
\item[] Algorithm: $\textit{parallel-optimize}$
\item[-] Input  : a statement $S$ of the language presented in Section~\ref{lang}
and a set of variables $l^\prime$ that we consider live (their
values concern us) at the end of executing $S$;
\item[-] Output: an optimized and may be corrected version $S^\prime$ of $S$
such that the relation between $S$ and $S^\prime$ is as stated
in~Theorem~\ref{sounddead}.
\item[-] Method  :\begin{enumerate}
    \item Find $\pts$ such that $S:\bot\ra \pts$ in the type system for pointer analysis.
    \item Find $l$ such that $S:(\bot,l)\ra (\pts,l^\prime)$ in
    the type system for live-variables analysis.
    \item Find $S^\prime$ such that $S:(\bot,l)\ra (\pts,l^\prime)\hookrightarrow
    S^\prime$ in the type system for dead code elimination.
\end{enumerate}
\end{itemize}
    \caption{The algorithm $\textit{optimize-parallel}$}\label{optimize}
}}\end{figure}

Figure~\ref{optimize} outlines an algorithm,
\textit{parallel-optimize}, that summarizes the optimization
process. A pointer analysis that annotates the points of the input
program with pointer information is the first step of the algorithm.
This step takes the form of a post type derivation of $S$, in our
type system for pointer analysis, using the bottom points-to type
$\bot=\{x\mapsto\emptyset\mid x\in \textit{Var}\}$ as the pre type.
Secondly, the algorithm refines the pointer information obtained in
the first step via annotating the pointer types with type components
for live-variables. Using our type systems for live-variables
analysis, this is done via a pre type derivation of $S$ for the set
$l^\prime$, the set of variables whose values concerns us at the end
of execution, as the post type. Finally, the information obtained so
far is utilized in the third step to find $S^\prime$ via using the
type system for dead code elimination proposed in this section.
Applying this algorithm to the program on the left-hand side of
Figure~\ref{example1} results in the program on the right-hand side
of the same figure. The details of this application is a simple
exercise.

The inference rules of our type system for dead code elimination are
as follows:  {\footnotesize {
\[
\begin{prooftree}
x \coloneqq e: \pts\ra \ptsp \quad x\notin l^\prime\justifies x
\coloneqq e: (\pts,l^\prime)\ra (\ptsp,l^\prime)\hookrightarrow
\textit{skip}\thickness=0.08em\using{(\coloneqq^e_1)}
\end{prooftree}\]\[ \begin{prooftree}
x \coloneqq e: \pts\ra \ptsp \quad x\in l^\prime\justifies x
\coloneqq e: (\pts,(l^\prime\setminus\{x\})\cup \textit{FV}(e))\ra
(\ptsp,l^\prime)\hookrightarrow x \coloneqq
e\thickness=0.08em\using{(\coloneqq^e_2)}
\end{prooftree}
\]\[
\begin{prooftree}
x\notin l^\prime\justifies x \coloneqq \& y: (\pts,l^\prime)\ra
(\pts[x\mapsto \{y^\prime\}],l^\prime)\hookrightarrow
\textit{skip}\thickness=0.08em \using{(\coloneqq
\&^e_1)}\end{prooftree}\quad
\begin{prooftree}
\justifies \textit{skip}: (\pts,l)\ra (\pts,l)\hookrightarrow
\textit{skip}\thickness=0.08em
\end{prooftree}
\]\[
\begin{prooftree} x\in l^\prime\justifies x \coloneqq \& y:
(\pts,l^\prime\setminus\{x\})\ra (\pts[x\mapsto
\{y^\prime\}],l^\prime)\hookrightarrow x \coloneqq \&
y\thickness=0.08em \using{(\coloneqq \&^e_2)}\end{prooftree}
\]\[
\begin{prooftree} x \coloneqq *y: \pts\ra
\ptsp\quad x\notin l^\prime\justifies x \coloneqq *y:
(\pts,l^\prime\cup\{y\})\ra (\ptsp,l^\prime)\hookrightarrow
\textit{skip}\thickness=0.08em \using{(\coloneqq *^e_1)}
\end{prooftree}
\]\[
\begin{prooftree} x \coloneqq *y: \pts\ra
\ptsp\quad x\in l^\prime\justifies x \coloneqq *y:
(\pts,(l^\prime\setminus\{x\})\cup\{y,z\mid z^\prime\in
\pts(y)\})\ra (\ptsp,l^\prime)\hookrightarrow x \coloneqq
*y\thickness=0.08em \using{(\coloneqq *^e_2)}
\end{prooftree}
\]\[
\begin{prooftree}
*x \coloneqq e: \pts\ra \ptsp \quad {\pts(x)\cap l}={\emptyset}
\justifies *x \coloneqq e: (\pts,l^\prime\cup\{x\})\ra
(\ptsp,l^\prime)\hookrightarrow
\textit{skip}\thickness=0.08em\using{(*\coloneqq^e_1)}
\end{prooftree}
\]\[
\begin{prooftree}
*x \coloneqq e: \pts\ra \ptsp \quad {\pts(x)\cap
l^\prime}\not={\emptyset} \justifies *x \coloneqq e:
(\pts,l^\prime\cup\{x\}\cup \textit{FV}(e))\ra
(\ptsp,l^\prime)\hookrightarrow *x \coloneqq
e\thickness=0.08em\using{(*\coloneqq^e_1)}
\end{prooftree}
\]\[
\begin{prooftree}
S_i:(\pts\cup \cup_{j\not=i} \pts_j,l_i)\ra (\pts_i,l^\prime\cup
\cup_{j\not=i} l_j)\hookrightarrow S_i^\prime\justifies
\textit{par}\{\{S_1\},\ldots,\{S_n\}\}: (\pts,\cup_i l_i)\ra (\cup_i
\pts_i,l^\prime)\hookrightarrow\textit{par}\{\{S_1^\prime\},
\ldots,\{S_n^\prime\}\}\thickness=0.08em
\using{(par^e)}\end{prooftree}
\]\[
\begin{prooftree}
\begin{tabular}{l}
$\textit{par}\{\{\mathit{if\ b_1\ then\ }S_1 \mathit{\ else\
skip}\},\ldots,\{\mathit{if\ b_n\ then\ }S_n \mathit{\
else\ skip}\}\}: (\pts,l)\ra (\ptsp,l^\prime)$\\
$\hookrightarrow \textit{par}\{\{\mathit{if\ b_1\ then\ }S_1^\prime
\mathit{\ else\ skip}\},\ldots,\{\mathit{if\ b_n\ then\ }S_n^\prime
\mathit{\ else\ skip}\}\}$
\end{tabular}
\justifies \begin{tabular}{l}
$\textit{par-if}\{(b_1,S_1),\ldots,(b_n,S_n)\}:
(\pts,l)\ra (\ptsp,l^\prime)$ \\
$\hookrightarrow\textit{par-if}\{(b_1,S_1^\prime),\ldots,(b_n,S_n^\prime)\}$
\end{tabular}
\thickness=0.08em \using{(\textit{par-if}^e)}\end{prooftree}
\]\[
\begin{prooftree}
S: (\pts\cup \ptsp,l)\ra (\ptsp,l^\prime\cup l)\hookrightarrow
S^\prime\justifies \textit{par-for}\{S\}: (\pts,l)\ra
(\ptsp,l^\prime)\hookrightarrow
\textit{par-for}\{S^\prime\}\thickness=0.08em
\using{(\textit{par-for}^e)}\end{prooftree}
\]\[
\begin{prooftree} S_1: (\pts,l) \ra (\pts^{\prime\prime},l^{\prime\prime})
\hookrightarrow S_1^\prime\quad S_2:
(\pts^{\prime\prime},l^{\prime\prime}) \ra
(\ptsp,l^\prime)\hookrightarrow S_2^\prime\justifies S_1;S_2:
(\pts,l)\ra (\ptsp,l^\prime)\hookrightarrow S_1^\prime;S_2^\prime
\thickness=0.08em\using{(seq^e)}
\end{prooftree}
\]\[
\begin{prooftree}
S_t:(\pts,l)\ra (\ptsp,l^\prime)\hookrightarrow S_t^\prime \quad
S_f:(\pts,l)\ra (\ptsp,l^\prime)\hookrightarrow S_f^\prime\justifies
\ifs: (\pts,l\cup \textit{FV}(b))\ra (\ptsp,l^\prime)\hookrightarrow
\mathit{if\ b\ then\ }S_t^\prime \mathit{\ else\
}S_f^\prime\thickness=0.08em\using{(\textit{if}^e)}
\end{prooftree}
\]\[
\begin{prooftree}
{l}={l^\prime\cup \textit{FV}(b)}\quad S_t: (\pts,l^\prime)\ra
(\pts,l)\hookrightarrow S_t^\prime \justifies \while: (\pts,l)\ra
(\pts,l^\prime)\hookrightarrow\mathit{while\ b\ do\ } S_t^\prime
\thickness=0.08em\using{(\textit{whl}^e)}
\end{prooftree}
\]\[
\begin{prooftree}
(\ptsp_1,l_1^\prime)\leq (\pts_1,l_1)\quad S:(\pts_1,l_1) \ra
(\pts_2,l_2)\hookrightarrow S^\prime \quad (\pts_2,l_2) \leq
(\ptsp_2,l_2^\prime)\justifies S: (\ptsp_1,l_1^\prime)\ra
(\ptsp_2,l_2^\prime)\hookrightarrow
S^\prime\thickness=0.08em\using{(\textit{csq}^e)}
\end{prooftree}
\]
}}

When optimizing programs it is important to guarantee that if (a)
the original and optimized programs are executed in similar states,
and (b) the original program ends at a state (rather than
\textit{abort}), then (a) the optimized program does not abort as
well, and (b) the optimized program reaches a state similar to that
reached by the original program. Indeed, this is guaranteed by the
following theorem.

\begin{theorem} \( (Soundness)\)\label{sounddead}
Suppose that $S:(\pts,l) \ra (\ptsp,l^\prime)\hookrightarrow
    S^\prime$ and $\gamma \thicksim_{(\pts,l)}\gamma_*$. Then\begin{enumerate}
\item If $S:\gamma\rightsquigarrow \gamma^\prime$, then there exists a state
$\gamma_*^\prime$ such that $S^\prime:\gamma_*\ra\gamma_*^\prime$
and $\gamma^\prime \thicksim_{(\ptsp,l^\prime)}\gamma_*^\prime$.
\item If $S^\prime:\gamma_*\ra\gamma_*^\prime$
and $S$ does not abort at $\gamma$, then there exists a state
$\gamma^\prime$ such that $S:\gamma\rightsquigarrow\gamma^\prime$
and $\gamma^\prime \thicksim_{(\ptsp,l^\prime)}\gamma_*^\prime$.
\end{enumerate}
\end{theorem}
The proof of this theorem is by induction on the structure of type
derivation and it follows smoothly
from~Theorem~\ref{soundness-live}. More
precisely~Theorem~\ref{soundness-live} is used when $S^\prime=S$.
When $S^\prime=\textit{skip}$, we take $\gamma_*^\prime=\gamma_*$ in
1. We note that the requirement of~Theorem~\ref{soundness-live} that
$S$ does not abort at $\gamma_*$ is guaranteed when this theorem is
called in the proof of Theorem~\ref{sounddead}.

\section{Related work}\label{rwork}
\subsubsection{Analysis of multithreaded programs:}
The analysis of multithreaded programs is an area that receives
growing interest. It is a challenging area~\cite{Sarkar09} as the
presence of threading complicates the program analysis. The work in
this area can be classified into two main categories. One category
includes techniques that were designed specifically to optimize or
correct multithreaded programs. The other category includes
techniques whose scope was extended from sequential programs to
multithreaded programs.

The first category mentioned above covers several directions of
research; synchronization analysis, deadlock, data race, and memory
consistency. The purpose in the analysis of synchronization
constructs ~\cite{Tian09,Xu10} is to clarify how the synchronization
actions apart executions of program segments. The result of this
analysis can be used by compiler to conveniently add join-fork
constructs. Clearly, adding such join-fork constructs will reduce
the run time of the program. One problem of multithreading computing
is deadlock which results from round waiting to gain resources.
Researchers have developed various techniques for deadlock
detection~\cite{Kim09,Wang08,Xiao10}. The situation when a memory
location is accessed by two threads (one of them writes in the
location) without synchronization is called data race. On direction
of research in this category focuses on data race
detection~\cite{Leung09}. The analysis of multithreaded programs
becomes even harder in the presence of a weak memory consistency
model because such model does not guarantee that a write statement
included in one thread is observed by other threads in the same
order. However such model simplifies some issues on the hardware
level. The work in this direction, like~\cite{Gelado10}, aims at
overcomes the drawbacks of using a simple consistency memory model.
Asymmetric Distributed Shared Memory (ADSM), a programming model
serving heterogeneous computing, is introduced in~\cite{Gelado10}.
ADSM manages a shared virtual memory to enable CPUs-access to
addresses on the accelerator real memory.

Under the second category mentioned above comes several directions
of research. One such direction is the using of flow-insensitive
analysis techniques to analyze multithreaded
programs~\cite{Lubbers08,Ruf00}. Although flow-insensitive
techniques are not very precise, some applications can afford that.
Examples of program analyses whose techniques were extended to cover
multithreaded programs are code motion~\cite{Knoop99}, constant
propagation~\cite{Lee98}, data flow for multithreaded programs with
copy-in and copy-out memory semantics~\cite{Kin03,Liu09}, and
concurrent static single assignment form~\cite{Lee97}.

The problem with almost all the work refereed to above is that it
does not apply to pointer programs. More precisely, for some of the
work the application is possible only if we have the result of a
pointer analysis for the input pointer program. The technique
presented in this paper for optimizing multithreaded programs has
the advantage of being simpler and more reliable than the
optimization techniques refereed to above that would work in the
presence of a pointer analysis.
\subsubsection{Pointer analysis:}
The pointer analysis for sequential programs has been studied
extensively for decades~\cite{Hind01}. One way to classify the work
in this area is according to properties of flow-sensitivity and
context-sensitivity. Hence the work is classified into
flow-sensitivity, flow-insensitivity, context-sensitivity, and
context-insensitivity.

Flow-sensitive analyses~\cite{Hardekopf09,Wang09,Yu10}, which are
more natural than flow-insensitive to most applications, consider
the order of program commands. Mostly these analyses perform an
abstract interpretation of program using dataflow analysis to
associate each program point with a points-to relation.
Flow-insensitive analyses~\cite{Adams02,Anderson02} do not consider
the order of program commands. Typically the output of these
analyses, which are performed using a constraint-based approach, is
a points-to relation that is valid all over the program. Clearly the
flow-sensitive approach is more precise but less efficient than the
flow-insensitive one. Moreover flow-insensitive techniques can be
used to analyze multithreaded programs.

The idea of context-sensitive approach~\cite{Nasre09,Yu10} is to
produce a points-to relation for the context of each call site of
each procedure. On the other hand, the
context-insensitive~\cite{Liang01} pointer analysis produces one
points-to relation for each procedure to cover contexts of all call
sites. As expected the context-sensitive approach is more precise
but less efficient than the context-insensitive one.

Although the problem of pointer analysis for sequential programs was
studied extensively, a little effort was done towards a pointer
analysis for multithreaded programs. In~\cite{Rugina03}, a flow
sensitive analysis for multithreaded programs was introduced. This
analysis associates each program point with a triple of points-to
relations. This in turn complicates the the analysis and creates a
sort of redundancy in the collected points-to information.
Investigating the details of this approach and our work makes it
apparent that our work is simpler and more accurate than this
approach. Moreover our approach provides a proof for the correctness
of the pointer analysis for each program. To the best of our
knowledge, such proof is not known to be provided by any other
existing approach.

\subsubsection{Type systems in program analysis:}
The work
in~\cite{Benton04,Laud06,Nielson02,El-Zawawy11,El-Zawawy11-2,El-Zawawy11-3}
is among the closest work to ours in the sense that it uses type
systems to achieve the program analysis in a way similar to the
present paper. The work in~\cite{Saabas08} can be seen as a special
case of our work for the case of while language where there is no
threading nor pointer constructs.

The work in~\cite{Laud06} shows that a good deal of program analysis
can be done using type systems. More precisely, it proves that for
every analysis in a certain class of data-flow analyses, there
exists a type system such that a program checks with a type if and
only if the type is a supertype for the set resulting from running
the analysis on the program. The type system in~\cite{Naik08} and
the flow-logic work in~\cite{Nielson02}, which is used
in~\cite{Nicola10} to study security of the coordinated systems, are
very similar to~\cite{Laud06}. Moreover, the work~\cite{Naik08}
transforms logical statements about programs to statements about the
program optimizations. For the simple while language, the work
in~\cite{Benton04} introduces type systems for constant folding and
dead code elimination and also logically proves correctness of
optimizations. The bidirectional data-flow analyses and their
program optimizations are treated with type systems
in~\cite{Frade09}. Earlier, related work (with structurally-complex
type systems) is~\cite{Palsberg95}. The work in~\cite{El-Zawawy11-2}
presents type systems that checks memory safety of multithreaded
programs using sensitive-nonsensitive pointer analysis.

To the best of our knowledge, our approach is the first attempt to
use type systems to optimize multithreaded programs and associates
every individual optimization with a justification for correctness.

\end{document}

